\documentclass{article}
\usepackage{graphicx} 

\usepackage{framed}

\usepackage[usenames,dvipsnames,svgnames,table]{xcolor}
\usepackage[english]{babel}
\usepackage[hidelinks]{hyperref}

\usepackage{amsmath,amssymb,amsthm}
\usepackage{mathtools}
\usepackage{graphicx}
\usepackage{enumitem}

\usepackage{palatino}
\usepackage{mathpazo}
\usepackage{fullpage}
\hypersetup{pdfpagemode=UseNone}
\usepackage[margin=1cm]{caption}
\usepackage[parfill]{parskip} 

\usepackage{thmtools}
\usepackage{thm-restate}
\usepackage[longend,ruled,linesnumbered,nokwfunc]{algorithm2e}
\SetKwInput{Input}{Input}
\SetKwInput{Output}{Output}

\usepackage[capitalize]{cleveref}

\newcommand{\ignore}[1]{}

\newcommand{\R}{\mathbb{R}}

\renewcommand{\max}{\mathrm{max}}
\renewcommand{\epsilon}{\varepsilon}
\newcommand{\eps}{\epsilon}

\newcommand{\Diag}{\text{\rm Diag}}
\def\01{\{0,1\}}

\newtheorem{defin}{Definition} 
\newtheorem{definition}[defin]{Definition}

\newtheorem{theorem}{Theorem}
\newtheorem*{theorem*}{Theorem}

\newtheorem{lemma}[defin]{Lemma}

\newtheorem*{claim*}{Claim}

\newtheorem*{conjecture*}{Conjecture}
\theoremstyle{definition}

\newcommand{\ma}{\mathbf{A}}
\newcommand{\mh}{\mathbf{H}}
\newcommand{\mv}{\mathbf{V}}
\newcommand{\mw}{\mathbf{W}}

\newcommand{\defeq}{:=}
\newcommand{\norm}[1]{\|#1\|}
\newcommand{\normFull}[1]{\left\|#1\right\|}
\newcommand{\tr}{\mathrm{tr}}
\newcommand{\vzero}{\mathbf{0}}
\newcommand{\tO}{\widetilde{O}}
\newcommand{\poly}{\mathrm{poly}}

\newcommand{\iter}{\mathcal{T}}

\newcommand{\sigmaipw}{\sigma_{i}(\mw^{\frac{1}{2}-\frac{1}{p}}\ma)}
\newcommand{\sigmapw}{\sigma(\mw^{\frac12-\frac1p} \ma)}

\newcommand{\sigmaipv}{\sigma_{i}(\mv^{\frac{1}{2}-\frac{1}{p}}\ma)}

\title{On computing approximate Lewis weights}
\author{Simon Apers\thanks{Universit\'e Paris Cit\'e, CNRS, IRIF, Paris, France. \texttt{apers@irif.fr}} \and Sander Gribling\thanks{Tilburg University, the Netherlands. \texttt{s.j.gribling@tilburguniversity.edu}} \and Aaron Sidford\thanks{Stanford University. \texttt{sidford@stanford.edu}}}

\begin{document}

\maketitle

\begin{abstract}
In this note we provide and analyze a simple method that given an $n \times d$ matrix, outputs approximate $\ell_p$-Lewis weights, a natural measure of the importance of the rows with respect to the $\ell_p$ norm, for $p \geq 2$. More precisely, we provide a simple post-processing procedure that turns natural one-sided approximate $\ell_p$-Lewis weights into two-sided approximations. When combined with a simple one-sided approximation algorithm presented by Lee (PhD thesis, `16) this yields an algorithm for computing two-sided approximations of the $\ell_p$-Lewis weights of an $n \times d$-matrix using $\poly(d,p)$ approximate leverage score computations. While efficient high-accuracy algorithms for approximating $\ell_p$-Lewis had been established previously by Fazel, Lee, Padmanabhan and Sidford (SODA `22), the simple structure and approximation tolerance of our algorithm may make it of use for different applications.
\end{abstract}

\section{Introduction}

For a matrix $\ma \in \R^{n \times d}$, its \emph{leverage scores} $\sigma(\ma) \in \R^n$, are a prominent importance measure of its rows.
The leverage score of row $i$, $\sigma_i(\ma) = [\sigma(\ma)]_i$ is defined as
\[
\sigma_i(\ma)
\defeq a_i^\top(\ma^\top \ma)^+ a_i,
\quad \forall i \in [n],
\]
where $a_i^\top$ is the $i$-th row of $A$ and $(\ma^\top \ma)^+$ is the pseudoinverse of $\ma^\top \ma$.
It holds that $\sum_i \sigma_i(\ma) \leq d$, with equality if $A$ has rank $d$.
Leverage scores have widespread applications in statistics \cite{rudelson2007sampling}, randomized linear algebra \cite{drineas2012fast}, and graph algorithms \cite{spielman2008graph}.

While leverage scores measure row importance with respect to the $\ell_2$-norm, there is a generalized notion called \emph{Lewis weights} that measures row importances with respect to the $\ell_p$-norm~\cite{Lewis78,BLM89,cohen2015lp}.
For finite $p > 0$, the $\ell_p$-Lewis weights of $\ma$ are defined as the unique vector $w \in \R^n_{> 0}$ that satisfies 
\begin{equation} \label{eq:Lewis}
w = \sigmapw, 
\end{equation}
We denote the $\ell_p$-Lewis weights of $\ma$ by $w_p(\ma)$, and note that $w_2(\ma) = \sigma(\ma)$.
Lewis weights are connected to computational geometry \cite{todd2016minimum}, and have applications in $\ell_p$-regression, e.g., \cite{cohen2015lp,JLS22}, and convex optimization, e.g.,  \cite{lee2019solving,BLLSSWW21}.

Lewis weights lack an explicit description, and computing them is somewhat more involved  than the $p=2$ case of leverage scores. For $p \in (0,4)$ there exists a (simple) iterative scheme, based on computing leverage scores~\cite{cohen2015lp}. For $p \geq 4$ the task is more challenging, see for example the recent works~\cite{cohen2015lp,lee2019solving,fazel22,Woodruff2023online}. 
Luckily, as for leverage scores, an \emph{$\eps$-estimate} of the form
\begin{equation} \label{eq:mult}
(1-\eps) w_p(\ma)
\leq v
\leq (1+\eps) w_p(\ma),
\end{equation}
where the inequality is elementwise,
often suffices in applications such $\ell_p$-embedding and -regression problems~\cite{talagrand1990embedding,cohen2015lp,Woodruff2023online}, approximating John ellipsoids~\cite{lee16}, or interior point methods~\cite{apers2023quantum}.

The implicit definition in \cref{eq:Lewis} suggests that alternative, even weaker notions of approximation may be of utility.
For example, we call a vector $w$ a \emph{two-sided $\eps$-approximation} of the Lewis weights if it approximately satisfies the fixed-point equation \eqref{eq:Lewis} in the following way: 
\begin{equation} \label{eq:two-sided}
(1-\eps) \sigmapw
\leq w
\leq (1+\eps) \sigmapw.
\end{equation}
It was shown in \cite[Lemma 14]{fazel22} that a two-sided $\eps$-approximation is also a $O(\eps p^2 \sqrt{d})$-estimate of the true Lewis weights.
An even weaker \emph{one-sided $\eps$-approximation} satisfies
\begin{equation} \label{eq:one-sided}
(1-\eps)\sigmapw \leq w
\quad \text{ and } \quad
\|w\|_1 \leq (1+\eps)d,
\end{equation}
and, as is argued in \cite{talagrand1990embedding,cohen2015lp}, one-sided approximation suffices for $\ell_p$-embedding and -regression problems, even when $\|w\|_1 = O(d)$.

Correspondingly, a simple algorithm to obtain such a one-sided approximation was proposed by Lee~\cite{lee16}; it allows one to construct a one-sided $\eps$-approximation by computing $\tO(1/\eps)$ many $O(\eps)$-approximate leverage scores.\footnote{In \label{footnote} \cite[Theorem~5.3.4]{lee16}, it is stated that the output of the algorithm even achieves a two-sided approximation, but the proof only establishes a one-sided approximation (and this suffices for the applications in \cite{lee16}).} 
Their algorithm is related to the aforementioned algorithm by Cohen and Peng \cite{cohen2015lp},
which for constant $0 < p < 4$ computes a stronger $\eps$-estimate of the $\ell_p$-Lewis weights using $\tO(1)$ many $O(\eps)$-approximate leverage score computations.
Crucially however, their algorithm is limited to $p < 4$, whereas many applications require larger (even non-constant) $p$.
More recently, Fazel, Lee, Padmanabhan and Sidford \cite{fazel22} provided an algorithm for computing an $\eps$-estimate of the $\ell_p$-Lewis weights for general~$p$, and to high precision, using $\tO(1)$ many \emph{exact} leverage score computations.
A naive analysis of this last algorithm suggests that approximate leverage scores suffice, albeit with a precision that is $\poly(\eps,1/d,1/n)$.
We generally have that $n \gg d$, and such inverse polynomial dependence on $n$ can be particularly prohibitive, ruling out the use of fast algorithms for approximating leverage scores. 

In this note we focus on the $p \geq 2$ case and address the gap between the one-sided approximation in~\cite{lee16}, requiring low-precision leverage scores, and the multiplicative Lewis weight estimates in~\cite{fazel22}, requiring precision inversely polynomial in $n$.
We show that it is possible to achieve two-sided and even multiplicative Lewis weight approximations by making $\poly(d,1/\eps)$ many $\poly(1/d,\eps)$-approximate leverage score computations, as is implied by the following theorem. 

\begin{theorem}[Low-precision Lewis weights] \label{thm:intro}
Let $\ma \in \R^{n \times d}$, $1 > \eps > 0$ and $p \geq 2$.
Using $\tO(p d/\eps)$ many $O(\eps/(pd))$-approximate leverage score computations it is possible to compute a two-sided $\eps$-approximation \eqref{eq:two-sided} of the $\ell_p$-Lewis weights of $\ma$.
\end{theorem}

Concretely, we show that running one ``fixed point iteration'' (as used e.g.~in the algorithm in~\cite{cohen2015lp}) can turn a one-sided $\eps$-approximation into a two-sided $O(p d \eps)$-approximation of the Lewis weights.
By using Lee's algorithm to obtain the one-sided approximation, we obtain the stated theorem.
By \cite[Lemma 14]{fazel22} this effectively also yields a $O(p^3 d^{3/2} \eps)$-estimate of the true Lewis weights. 

Recently, two of the authors provided an improved quantum algorithm for computing approximate leverage scores, see~\cite{apers2023quantum}. This led to quantum speedups for linear programming by using it as a subroutine in a dual-only interior point method based on either the volumetric or Lewis weight barrier. For the latter, multiplicative estimates of $w_p(\ma)$ are needed; these are obtained through a quantum implementation of the algorithm presented in this note.

\section{Preliminaries}

Throughout, we use plain italic for vectors $v$ and boldface capitals $\ma$ for matrices.
If there can be no confusion, we use a boldface capital to denote the diagonal matrix corresponding to its lowercase counterpart, e.g., $\mw = \Diag(w)$. We use $\vzero$ to denote the all-zero vector of appropriate size.
For vectors $u,v \in \R^n$ we write $u \leq v$ as shorthand for the entrywise inequality $u_i \leq v_i$ for all $i \in [n]$. We use ``$\preceq$'' and ''$\succeq$'' to denote the Loewner order on symmetric matrices, that is, we write $A \succeq B$ or $B \preceq A$ when $A-B$ is positive semidefinite. 

We formally define the following three notions of approximate Lewis weights.

\vspace{2mm}
\begin{definition}[Lewis weight approximations]
Let $\ma \in \R^{n \times d}$, $w \in \R^n_{> 0}$, $1 > \eps > 0$ and $p \geq 2$.
\begin{itemize}
\item
$w$ is a one-sided $\epsilon$-approximate $\ell_p$-Lewis weight of $\ma$ if $\sigmapw \leq (1+\eps) w$ and $\| w_i \|_1 \leq (1+\eps)d$. 
\item 
$w$ is a two-sided $\epsilon$-approximate $\ell_p$-Lewis weight of $\ma$ if $(1-\eps)  \sigmapw \leq w \leq (1+\eps) \sigmapw$.
\item
$w$ is an \emph{$\eps$-estimate} of $w_p(\ma)$ if $(1-\eps)  w_p(\ma) \leq w \leq (1+\eps) w_p(\ma)$.
\end{itemize}
\end{definition}
Note that when $p=2$, the case of leverage scores, the last two notions of approximation coincide.

Next, we prove two simple technical lemmas that we use.
\begin{lemma} \label{lem:abs_diff_2}
    Suppose $x,y \in \R^n_{\geq 0}$ and $\delta>0$ are such that $y \leq (1+\delta) x$ entrywise and $\|x\|_1 \leq (1+\delta)\|y\|_1$. Then $\|x-y\|_1 \leq 3 \delta \|y\|_1$.
\end{lemma}
\begin{proof}
The proof follows from writing $x-y = (x-\frac{1}{1+\delta} y) - \frac{\delta}{1+\delta} y$ and applying the triangle inequality:
\begin{align*}
\norm{x-y}_{1} &\leq\sum_{i\in[n]}\left|x_{i}-\frac{1}{1+\delta}y_{i}\right|+ \frac{\delta}{1+\delta} \sum_{i\in[n]}|y_{i}| =\sum_{i\in[n]}x_{i}-\frac{1}{1+\delta}y_{i}+ \frac{\delta}{1+\delta} \sum_{i\in[n]}y_{i} \\
&= \|x\|_1 - \frac{1}{1+\delta} \|y\|_1 + \frac{\delta}{1+\delta} \|y\|_1. 
\end{align*}
Finally, using $\|x\|_1 \leq (1+\delta) \|y\|_1$ we obtain $\norm{x-y}_{1} \leq \left(\frac{(1+\delta)^2 - 1+\delta}{1+\delta}\right) \|y\|_1 \leq 3 \delta \|y\|_1$.
\end{proof}

\begin{lemma}\label{lemma:1mxc_upper}
$|1 - x^c| \leq c \cdot \max\{1,x\}^{c - 1} \cdot | 1 - x|$ for all $x \geq 0$ and $c \geq 1$. 
\end{lemma}

\begin{proof} As the claim is trivial when $c = 1$, we assume $c > 1$. In this case,
\[
|1 - x^c| = \left| \int_1^{x} c \cdot  \alpha^{c - 1} d \alpha \right|
\leq 
\left| \int_1^{x} c\cdot  \max\{1,x\}^{c - 1} d\alpha \right|
= c \cdot \max\{1,x\}^{c - 1} \cdot | 1 - x|\,. \qedhere
\]
\end{proof}

\section{Main result} \label{sec:main}

As our main result, we will prove that if $w$ is a one-sided $\epsilon$-approximate $\ell_p$-Lewis weight of $\ma$, then the vector $\iter_p(w) \in \R^{n}$ defined coordinate-wise by 
\begin{equation} \label{eq:fp-iteration}
\iter_p(w)_i \defeq  w_i \cdot (\sigmaipw/w_{i})^{\frac{p}{2}}
= w_i^{1-\frac{p}{2}} \sigmaipw^{\frac{p}{2}} \qquad \text{ for all } i \in [n]
\end{equation}
is a two-sided $O(\eps p d)$-approximation of the $\ell_p$-Lewis weights of $\ma$.
This transformation effectively corresponds to running one fixed point iteration of the Lewis weight algorithm used in e.g., \cite{cohen2015lp}. 

We first prove a key lemma that allows us to compare quadratic forms associated with $\iter_p(w)$ and~$w$.

\begin{lemma}
\label{lem:mult_diff}
For any $w\in\R_{>0}^{n}$ and 
$\alpha = \big\|\iter_p(w) - \sigmapw \big\|_1$, we have
\[
(1-\alpha)\ma^{\top}\mw^{1-\frac{2}{p}}\ma
\preceq \ma^{\top}\Diag(\iter_p(w))^{1-\frac{2}{p}}\ma
\preceq (1+\alpha)\ma^{\top}\mw^{1-\frac{2}{p}}\ma
\,.
\]
\end{lemma}
\begin{proof}
Let $\Delta\defeq \iter_p(w)^{1-\frac{2}{p}}-w^{1-\frac{2}{p}}$ and let $\Delta_{+}\defeq\max\{\Delta,\vzero\}$
and $\Delta_{-}\defeq\max\{-\Delta,\vzero\}$ entrywise so that $\Delta_{+},\Delta_{-}\in\R_{\geq0}^{n}$,
$\Delta=\Delta_{+}-\Delta_{-}$, and $\norm{\Delta}_{1}=\norm{\Delta_{+}}_{1}+\norm{\Delta_{-}}_{1}$.
Letting $\mh\defeq\ma^{\top}\mw^{1-\frac{2}{p}}\ma$ we have that
\begin{align*}
\normFull{\mh^{-1/2}\left(\ma^{\top}\Delta\ma\right)\mh^{-1/2}}_{2} & \leq\normFull{\mh^{-1/2}\left(\ma^{\top}\Delta_{+}\ma\right)\mh^{-1/2}}_{2}+\normFull{\mh^{-1/2}\left(\ma^{\top}\Delta_{-}\ma\right)\mh^{-1/2}}_{2}\,.
\end{align*}
We then use the fact $\Delta_+$ is positive semidefinite, and therefore so is $\mh^{-1/2}\left(\ma^{\top}\Delta_{+}\ma\right)\mh^{-1/2}$, to upper bound the spectral norm by the trace: 
\begin{align*}
\normFull{\mh^{-1/2}\left(\ma^{\top}\Delta_{+}\ma\right)\mh^{-1/2}}_{2} & \leq\tr\left[\mh^{-1/2}\left(\ma^{\top}\Delta_{+}\ma\right)\mh^{-1/2}\right]=\sum_{i\in[n]}\left[\Delta_{+}\right]_{i}\left[\ma(\ma^{\top}\mw^{1-\frac{2}{p}}\ma)^{-1}\ma^{\top}\right]_{ii}\\
 & =\sum_{i\in[n]}\left[\Delta_{+}\right]_{i}\cdot\frac{\sigmaipw}{w_{i}^{1-\frac{2}{p}}}\,.
\end{align*}
By symmetry the same bound holds for $\Delta_{-}$.
Combining the two bounds yields that 
\[
\normFull{\mh^{-1/2}\left(\ma^{\top}\Delta\ma\right)\mh^{-1/2}}_{2}\leq\sum_{i\in[n]}\left|\Delta_{i}\right|\cdot\frac{\sigmaipw}{w_{i}^{1-\frac{2}{p}}}
\]
Recalling that $\Delta = \iter_p(w)^{1-\frac2p} - w^{1-\frac2p}$ we obtain the desired upper bound:
\[
\sum_{i\in[n]}\left|\Delta_{i}\right|\cdot\frac{\sigmaipw}{w_{i}^{1-\frac{2}{p}}}
=
\big\| \iter_p(w) - \sigmapw 
 \big\|_1\,. \qedhere
\]
\end{proof}

Next, we show how to bound the approximation factor $\alpha = \big\|\iter_p(w) - \sigmapw \big\|_1$ when $w$ is a one-sided approximation of the Lewis weights of $\ma$.

\begin{lemma}\label{lem:distances} 
If $w \in \R^n_{> 0}$ is a one-sided $\epsilon$-approximate $\ell_p$-Lewis weight of $\ma$ for  $p \geq 2$, then 
\[
\norm{\iter_p(w) - \sigmaipw}_1 \leq 3\epsilon (\frac{p}{2} - 1) (1 + \epsilon)^{\frac{p}{2}-1} d.
\]
\end{lemma}

\begin{proof}
Note that $\iter_p(w)_i = \sigmaipw (\sigmaipw / w_i )^{\frac{p}{2}-1}$. Since $\iter_p(w)$ and $\sigmapw$ are non-negative, it follows that
\begin{align*}
\norm{\iter_p(w) - \sigmapw}_1
&= \sum_{i \in [n]} \sigmaipw \left| \left(\frac{ \sigmaipw }{ w_i } \right)^{\frac{p}{2} - 1} - 1 \right|
\\
&\leq 
(\frac{p}{2} - 1)(1+ \epsilon)^{\frac{p}{2} - 1}
\norm{w - \sigmapw}_1
\, ,
\end{align*}
where in the last step we applied \Cref{lemma:1mxc_upper} and that $\sigmapw \leq (1+\epsilon) w$. The result then follows from the fact that $\norm{w - \sigmapw}_1 \leq 3\epsilon \norm{\sigmapw}_1 = 3 \epsilon d$, by \cref{lem:abs_diff_2}.
\end{proof}

We are now ready to prove our main technical theorem.

\begin{theorem}[1- to 2-sided]
\label{thm:main_thm} Let $w$ be a one-sided $\epsilon$-approximate $\ell_p$-Lewis weight of $\ma$ and let
\[
v = \iter_p(w) = w^{1-\frac{p}{2}} \sigmapw^{\frac{p}{2}}.
\]
Then, for $\alpha = 3\epsilon (\frac{p}{2} - 1) (1 + \epsilon)^{\frac{p}{2}-1} d$, we have 
$(1-\alpha)\sigmaipv\leq v_{i}\leq(1+\alpha)\sigmaipv$.
\end{theorem}

\begin{proof}
We have that
\begin{align*}
\frac{\sigma(\mv^{\frac{1}{2}-\frac{1}{p}}\ma)_{i}}{v_{i}} & =v_{i}^{-\frac{2}{p}}\cdot\left[\ma(\ma^{\top}\mv^{1-\frac{2}{p}}\ma)^{-1}\ma^{\top}\right]_{ii}=\frac{\left[\ma(\ma^{\top}\mv^{1-\frac{2}{p}}\ma)^{+}\ma^{\top}\right]_{ii}}{\left[\ma(\ma^{\top}\mw^{1-\frac{2}{p}}\ma)^{+}\ma^{\top}\right]_{ii}}\,.
\end{align*}
We wish to lower and upper bound this fraction by $1/(1+\alpha)$ and $1/(1-\alpha)$, respectively.
For this it suffices to prove that
\[
\frac{1}{1+\alpha} \left( \ma^{\top}\mw^{1-\frac{2}{p}} \ma \right)^+
\preceq \left( \ma^{\top}\mv^{1-\frac{2}{p}}\ma \right)^+
\preceq \frac{1}{1-\alpha} \left( \ma^{\top}\mw^{1-\frac{2}{p}} \ma \right)^+,
\]
which is equivalent to showing that
\[
(1-\alpha) \ma^{\top}\mw^{1-\frac{2}{p}}\ma 
\preceq \ma^{\top}\mv^{1-\frac{2}{p}}\ma
\preceq (1+\alpha) \ma^{\top}\mw^{1-\frac{2}{p}}\ma.
\]
By \cref{lem:mult_diff} and \cref{lem:distances} 
 this holds for $\alpha = \| \iter_p(w) - \sigmapw \|_1 \leq 3\epsilon (\frac{p}{2} - 1) (1 + \epsilon)^{\frac{p}{2}-1} d$.
\end{proof}

For $0< \eps \leq 1/p$, the approximation factor in \cref{thm:main_thm} is $\alpha = O(\eps p d)$.

\section{Obtaining multiplicative estimates of Lewis weights}
\label{sec:alg}

Through \cref{thm:main_thm}, we can show that a variation on Lee's algorithm can be used to compute multiplicative $\eps$-approximations of the Lewis weights using approximate leverage score computations at the expense of a $\poly(d,p)$-overhead in precision, see \cref{thm:algo}. 

We first prove a simple lemma that shows that two-sided approximation is ``stable'' with respect to a multiplicative change (i.e., if $v$ is a two-sided approximation then so is multiplicative approximation of it).  

\begin{lemma} \label{lem:stab}
    Let $\gamma \geq 1$.
    Let $v,\tilde v \in \R^n_{> 0}$ be such that $\gamma^{-1} \tilde v_i \leq v_i \leq \gamma \tilde v_i$ for all $i \in [n]$. Then 
    \[
    \gamma^{-1} \frac{\sigma_i(\widetilde \mv^{\frac{1}{2}-\frac{1}{p}}\ma)}{\tilde v_i} \leq \frac{\sigma_i(\mv^{\frac{1}{2}-\frac{1}{p}}\ma)}{v_i} \leq \gamma \frac{\sigma_i(\widetilde \mv^{\frac{1}{2}-\frac{1}{p}}\ma)}{\tilde v_i}
    \]
\end{lemma}
\begin{proof}
We first prove the first inequality. We have that 
        \begin{align*}
        \frac{\sigma_i(\widetilde \mv^{\frac{1}{2}-\frac{1}{p}}\ma)}{\tilde v_i} &= \tilde v_{i}^{-\frac{2}{p}}\cdot\left[\ma(\ma^{\top} \widetilde \mv^{1-\frac{2}{p}}\ma)^{+}\ma^{\top}\right]_{ii}
        \leq \gamma \, v_{i}^{-\frac{2}{p}}\cdot\left[\ma(\ma^{\top}  \mv^{1-\frac{2}{p}}\ma)^{+} \ma^{\top}\right]_{ii}
        =\gamma \frac{\sigma_i( \mv^{\frac{1}{2}-\frac{1}{p}}\ma)}{v_i}
        \end{align*}
        where the inequality uses the estimates $\tilde v_i^{-\frac2p} \leq \gamma^{\frac2p} v_i^{-\frac2p}$ and $(\ma^\top \widetilde \mv^{1-\frac2p} \ma)^+ \preceq \gamma^{1-\frac{2}{p}} (\ma^\top \mv^{1-\frac2p} \ma)^+$. The second inequality of the lemma follows by exchanging the roles of $v$ and $\tilde v$. 
\end{proof}

We can now state our algorithm and prove its correctness. 

\begin{algorithm}[ht]
\caption{Two-sided Lewis weight approximation} \label{alg:low-precision-Lewis}
\Input{$\ma \in \R^{n \times d}$, accuracy $0<\alpha<1$, $p \geq 2$}
\Output{A vector $v \in \R^n$}

\BlankLine
Let $w_i^{(1)} = d/n$ for all $i \in [n]$, $\eps_1 = \alpha/(100 p d)$, $\eps_2 = \alpha/(3p)$, $T = \lceil 2\log(n/d)/\eps_1 \rceil$;

\For{$k=1,\ldots,T-1$}{ 
    Let $w^{(k+1)}$ be $\eps_1/4$-estimates of $\sigma((\mw^{(k)})^{\frac12 - \frac1p} \ma)$\;
}

Let $w  = \frac{1}{T} \sum_{k=1}^T w^{(k)}$ and $s$ be $\eps_2$-estimates of $\sigmapw$\;

\Return{$(v_i)_i$ with $v_i  = w_i (s_i/w_i)^{\frac{p}{2}}$ for all $i \in [n]$}
\end{algorithm}

\begin{theorem}[Approximate Lewis weights from approximate leverage scores] \label{thm:algo}
For any matrix $\ma \in \R^{n \times d}$, accuracy $0<\alpha<1$ and $p \geq 2$, \cref{alg:low-precision-Lewis} outputs a two-sided $\alpha$-approximation of the $\ell_p$-Lewis weights of $\ma$.
\end{theorem}
\begin{proof}
Steps 1.-4.~of the algorithm correspond to Algorithm 6 by Lee \cite{lee16}.
In \cite[Theorem~5.3.4]{lee16} it is shown that the resulting $w$ satisfies $w_i/\sigmaipw \geq \exp(-\eps_1)$ and hence $\sigmaipw \leq \exp(\eps_1) w_i \leq (1+2\eps_1) w_i$.
Moreover, $w$ is an average over $\eps_1/4$-approximate leverage scores so that $\|w\|_1 \leq (1+\eps_1/4) d$, and hence $w$ is a one-sided $2\eps_1$-approximation of the Lewis weights.
By our \cref{thm:main_thm} this implies the vector $\iter_p(w)$ is a two-sided Lewis weight approximation with approximation factor 
\[
6\epsilon_1 (\frac{p}{2} - 1) (1 + 2\epsilon_1)^{\frac{p}{2}-1} d  \leq \alpha/3
\]
by our choice of $\eps_1$.
Finally, we use $\eps_2$-estimates $s$ of $\sigmapw$ to define $v_i = w_i (s_i/w_{i})^{\frac{p}{2}}$, so that
\[
(1-\eps_2)^{p/2} \iter_p(w)_i
\leq v_i
\leq (1+\eps_2)^{p/2} \iter_p(w)_i
\leq \frac{1}{(1-\eps_2)^{p/2}} \iter_p(w)_i.
\]
We can now apply \cref{lem:mult_diff} with $\gamma = 1/(1-\eps_2)^{p/2} \leq 1 + \alpha/3$ by our choice of $\eps_2$.
This implies that the~$v_i$'s are two-sided Lewis weight approximations satisfying
\[
(1 - \alpha/3)^2
\leq \frac{\sigma_i(\mv^{\frac12 - \frac1p} \ma)}{v_i}
\leq (1 + \alpha/3)^2.
\]
Using that $(1-\alpha/3)^2 \geq 1-\alpha$ and $(1+\alpha/3)^2 \leq 1+\alpha$, this proves the claim.
\end{proof}

This theorem implies that we can obtain two-sided $\alpha$-approximations of the Lewis weights by iteratively computing $\tO(p d/\alpha)$ many $O(\alpha/(pd))$-approximate leverage scores.
This proves \cref{thm:intro} stated in the introduction.

\section*{Acknowledgements}

Simon Apers was partially supported by French projects EPIQ (ANR-22-PETQ-0007), QUDATA (ANR-18-CE47-0010) and QUOPS (ANR-22-CE47-0003-01), and EU project QOPT (QuantERA ERA-NET Cofund 2022-25). Aaron Sidford was supported in part by a Microsoft Research Faculty Fellowship, NSF CAREER Award CCF-1844855, NSF Grant CCF-1955039, and a PayPal research award.

\bibliographystyle{alpha}
\bibliography{biblio}

\end{document}